\documentclass[11pt]{article}

\usepackage[margin=1in]{geometry}
\usepackage{graphicx}
\usepackage[pdftex,colorlinks=true,linkcolor=blue,citecolor=blue,urlcolor=black]{hyperref}
\usepackage{array}
\usepackage{amsmath, amsthm, amssymb}
\usepackage{subfigure}
\usepackage{comment}
\usepackage{url}
\usepackage{pdflscape}
\input{Qcircuit}
\newcommand{\F}{\mathbb{F}}

\newcommand{\bracket}[3]{\langle #1|#2|#3 \rangle}
\newcommand{\braket}[1]{\langle #1 \rangle}
\newcommand{\sm}[1]{\left( \begin{smallmatrix} #1 \end{smallmatrix} \right)}

\DeclareMathOperator{\poly}{poly}

\DeclareMathOperator{\tr}{tr}

\DeclareMathOperator{\gap}{gap}

\newcommand{\be}{\begin{equation}}
\newcommand{\ee}{\end{equation}}
\newcommand{\bea}{\begin{eqnarray}}
\newcommand{\eea}{\end{eqnarray}}
\newcommand{\bes}{\begin{equation*}}
\newcommand{\ees}{\end{equation*}}
\newcommand{\beas}{\begin{eqnarray*}}
\newcommand{\eeas}{\end{eqnarray*}}

\def\clap#1{\hbox to 0pt{\hss#1\hss}}

\makeatletter
\newtheorem*{rep@theorem}{\rep@title}
\newcommand{\newreptheorem}[2]{%
\newenvironment{rep#1}[1]{%
 \def\rep@title{#2 \ref{##1} (restated)}%
 \begin{rep@theorem}}%
 {\end{rep@theorem}}}
\makeatother

\newtheorem{thm}{Theorem}
\newtheorem*{thm*}{Theorem}
\newtheorem{cor}[thm]{Corollary}

\newtheorem*{lem*}{Lemma}
\newtheorem{prop}[thm]{Proposition}

\newtheorem{obs}[thm]{Observation}

\newtheorem{prob}[thm]{Problem}

\newreptheorem{thm}{Theorem}
\newreptheorem{lem}{Lemma}

\begin{document}

\title{Quantum circuits and low-degree polynomials over $\F_2$}

\author{Ashley Montanaro\footnote{School of Mathematics, University of Bristol, UK; {\tt ashley.montanaro@bristol.ac.uk}.}}

\maketitle

\begin{abstract}
In this work we explore a correspondence between quantum circuits and low-degree polynomials over the finite field $\F_2$. Any quantum circuit made up of Hadamard, Z, controlled-Z and controlled-controlled-Z gates gives rise to a degree-3 polynomial over $\F_2$ such that calculating quantum circuit amplitudes is equivalent to counting zeroes of the corresponding polynomial. We exploit this connection, which is especially clean and simple for this particular gate set,  in two directions. First, we give proofs of classical hardness results based on quantum circuit concepts. Second, we find efficient classical simulation algorithms for certain classes of quantum circuits based on efficient algorithms for classes of polynomials.
\end{abstract}

% ------------------------------------------------------------------------------

\section{Introduction}

Quantum computers are believed to outperform classical computers for important tasks as varied as simulation of quantum mechanics and factorisation of large integers. Although no large-scale general-purpose quantum computer has been built as yet, quantum computation can nevertheless already be used as a theoretical tool to study other areas of science and mathematics, without the need for an actual quantum computer.

This work explores a simple correspondence between quantum circuits and low-degree polynomials over the finite field $\F_2$, i.e.\ the integers modulo 2. By picking the right gate set, it turns out that quantum circuit amplitudes have a close connection to counting zeroes of such polynomials. This correspondence can be exploited in two directions. On the one hand, ideas about quantum circuits can be used to prove purely classical results regarding the computational complexity of counting zeroes of polynomials over finite fields. On the other, known classical results about polynomials can be used to give new algorithms for simulating classes of quantum circuits.

A similar perspective has been taken by a number of previous works. Particularly relevant is prior work of Dawson et al.~\cite{dawson05}, who showed that quantum circuit amplitudes for circuits of Toffoli and Hadamard gates can be understood in terms of solutions to systems of polynomial equations involving low-degree polynomials over $\F_2$. Here we use a slightly different universal gate set: Hadamard ($=\frac{1}{\sqrt{2}}\sm{1 & 1\\1& -1}$), Z ($=\sm{1&0\\0&-1}$), controlled-Z (``CZ'') and controlled-controlled-Z (``CCZ''). This is essentially equivalent to the gate set of~\cite{dawson05}, as Toffoli gates are identical to CCZ gates conjugated by a Hadamard gate on the target qubit. However, this small shift in perspective seems to simplify and clarify some of the arguments involved. For example, the connection we use associates a single polynomial with each circuit. Related ideas to~\cite{dawson05} were used by Rudolph~\cite{rudolph09} to give a simple encoding of quantum circuit amplitudes as matrix permanents. The set of circuits we consider is a very special case of the class of ``algebraic quantum circuits'' studied by Bacon, van Dam and Russell~\cite{bacon08} in some generality.

The idea of proving classical results using quantum methods has also been explored previously; see~\cite{drucker11} for a survey of many results in this area. Within computational complexity alone, three relevant examples are Aaronson's proof of the computational hardness of computing the matrix permanent using the close connection between the permanent and linear-optical quantum circuits~\cite{aaronson11a}; Kuperberg's proof of the computational hardness of approximately computing Jones polynomials by expressing these in terms of quantum circuits~\cite{kuperberg15}; and Fujii and Morimae's proof of hardness of computing Ising model partition functions, again based on quantum circuits over a suitable gate set~\cite{fujii13}. More recently, together with Bremner and Shepherd~\cite{bremner15}, the present author used a correspondence between low-degree polynomials and a certain class of simple quantum computations, known as IQP circuits~\cite{shepherd09}, to argue that random IQP circuits are unlikely to be efficiently simulable classically. This holds even if the classical simulator is allowed to be approximate, with a fairly generous notion of approximation.

The correspondence between low-degree polynomials and quantum circuits which we investigate here seems particularly simple and direct. We have therefore tried to use it to highlight some of the beautiful ideas present in previous works, and to produce an accessible introduction to computational complexity issues suitable for physicists; and also an introduction suitable for computer scientists to how one can prove classical results using the quantum circuit model.

We begin by introducing the circuit-polynomial correspondence and proving its correctness, and go on to make some simple observations about this connection. Then, in Section \ref{sec:compcomp}, we introduce the ideas from computational complexity that we will need, and in Section \ref{sec:phard} show that the correspondence can be used to prove classical hardness of exactly computing the number of zeroes of low-degree polynomials. Similarly, in Section \ref{sec:approxcomplexity} we show that approximate computation of this quantity is closely related to quantum computation. We study a new complexity measure for polynomials motivated by this correspondence -- the quantum circuit width -- in Section \ref{sec:width}. Then, in Section \ref{sec:polysim}, we use the circuit-polynomial correspondence to give two simple classical simulation algorithms for classes of quantum circuits: circuits with few CCZ gates (or where the degree-3 part of the polynomial corresponding to the circuit has a small ``hitting set'', qv), and circuits whose corresponding polynomial can be simplified by a linear transformation. We conclude in Section \ref{sec:conclusions} with some open problems.

% ------------------------------------------------------------------------------

\section{Circuits and polynomials}
\label{sec:circpoly}

In this work, we consider quantum circuit amplitudes of the form $\braket{0|C|0}$, where $C$ is a unitary operator expressed as a circuit on $\ell$ qubits with $\poly(\ell)$ gates, and we write $\ket{0} = \ket{0}^{\otimes \ell}$ for conciseness throughout. The gates in $C$ are picked from the set $\mathcal{F} = \{$Hadamard, Z, CZ, CCZ$\}$\footnote{In fact, the Z and CZ gates are not necessary, as they can be produced from CCZ gates together with the use of ancillas, but it will be convenient to include them.}. Using the gate set $\mathcal{F}$ will allow us to write $\bracket{0}{C}{0}$ in a particularly concise form. Assume that $C$ begins and ends with a column of Hadamards, i.e.\ is of the form
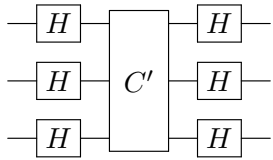
\[
\Qcircuit @C=1em @R=.7em {
 & \gate{H} & \multigate{2}{C'} & \gate{H} & \qw  \\
 & \gate{H} & \ghost{C'} & \gate{H} & \qw \\
 & \gate{H} & \ghost{C'} & \gate{H} & \qw 
}
\]
for some circuit $C'$. This is without loss of generality, as we can always add pairs of Hadamards to the beginning or end of each line without changing the unitary operator corresponding to the circuit. Further assume that $C'$ contains at least one gate acting on each qubit. Let $h$ be the number of internal Hadamard gates that $C$ contains, i.e.\ the number of Hadamards in $C'$. Set $n = h + \ell$ and define a polynomial $f_C:\{0,1\}^n \rightarrow \{0,1\}$ over $\F_2$ as follows. Divide each horizontal wire of the internal part $C'$ into segments, with each segment corresponding to a portion of the wire which is either between two Hadamard gates or to the left/right of all the Hadamard gates. Associate a distinct variable $x_i$ with each segment of each wire. Observe that there are exactly $h+\ell$ variables in total. Each Hadamard gate now joins two segments and associates their corresponding variables, and each Z, CZ, CCZ gate is associated with one, two or three (respectively) variables, corresponding to the segments on which it acts. For each set of variables $x_{i_1},\dots,x_{i_k}$ associated with each gate, add the corresponding term $x_{i_1} \dots x_{i_k}$ to $f_C$. As we are working over $\F_2$, all addition and multiplication in $f_C$ is taken modulo 2. Note that this procedure never produces polynomials of degree higher than 3.

As a simple example of this construction, consider the labelled circuit $C'$ in Figure \ref{fig:labcirc}, where we use the notation
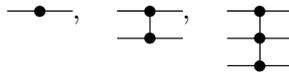
\[ \Qcircuit @C=1em @R=.7em { & \control \qw & \qw },\;\;\;\; \Qcircuit @C=1em @R=.7em { & \ctrl{1} \qw & \qw\\ & \control \qw & \qw },\;\;\;\; \Qcircuit @C=1em @R=.7em { & \ctrl{2} \qw & \qw\\ & \control \qw & \qw \\  & \control \qw & \qw } \]
for Z, CZ, CCZ gates respectively.
%
%\[
%\Qcircuit @C=1em @R=.7em {
% & \gate{H} & \ctrl{1} & \qw & \ctrl{2} & \gate{H} & \qw  \\
% & \qw & \control \qw & \gate{H} & \control \qw & \qw & \qw\\
% & \gate{H} & \qw & \control \qw & \control \qw & \qw & \qw
%}
%\]
%
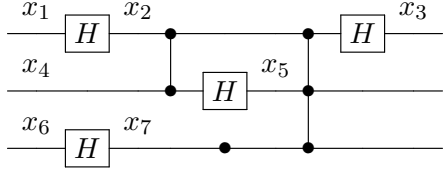
\begin{figure}
\[
\Qcircuit @C=1em @R=.7em {
& \ustick{x_1} \qw & \gate{H} & \ustick{x_2} \qw & \ctrl{1} & \qw & \qw & \ctrl{2} & \gate{H} & \ustick{x_3} \qw & \qw \\
& \ustick{x_4} \qw & \qw & \qw & \control \qw & \gate{H} & \ustick{x_5} \qw & \control \qw & \qw & \qw & \qw \\
& \ustick{x_6} \qw & \gate{H} & \ustick{x_7} \qw & \qw & \control \qw & \qw & \control \qw & \qw & \qw & \qw
}
\]
\caption{The internal part $C'$ of a circuit $C$ corresponding to the polynomial $x_1 x_2 + x_2 x_3 + x_4 x_5 + x_6 x_7 + x_2 x_4 + x_2 x_5 x_7 + x_7$.}
\label{fig:labcirc}
\end{figure}

We now show that the number of zeroes of the polynomial corresponding to $C$ has a close connection to $\braket{0|C|0}$. To be more precise, $\braket{0|C|0}$ is proportional to $\gap(f_C)$, where the gap of a polynomial is the difference between the number of zeroes and ones of that polynomial:
\[ \gap(f_C) := \sum_{x \in \{0,1\}^n} (-1)^{f_C(x)} = |\{x:f_C(x)=0\}|-|\{x:f_C(x)=1\}|. \]
A similar result was shown in~\cite{dawson05} for circuits containing Hadamard and Toffoli gates. However, the argument here seems somewhat simpler. Although there are several ways that the following result can be proven, we choose to highlight a connection to the beautiful results of~\cite{bremner11}.

\begin{prop}
\label{prop:gap}
Let $C$ be a quantum circuit on $\ell$ qubits consisting of Hadamard, Z, CZ and CCZ gates, starting and ending with a column of Hadamard gates, and containing $h$ internal Hadamard gates. Then
\[ \bracket{0}{C}{0} = \frac{\gap(f_C)}{2^{h/2+\ell}}. \]
\end{prop}

\begin{proof}
First consider the case where the internal part $C'$ of $C$ does not contain any Hadamard gates (as treated in~\cite[Appendix B]{bremner15}). Let $Z_i$ denote a Z gate acting on the $i$'th qubit (and similarly $CZ_{ij}$, $CCZ_{ijk}$). Then, for any $x \in \{0,1\}^\ell$, $\braket{x|Z_i|x} = (-1)^{x_i}$, $\braket{x|CZ_{ij}|x} = (-1)^{x_i x_j}$, $\braket{x|CCZ_{ijk}|x} = (-1)^{x_i x_j x_k}$. As these gates are diagonal, we can obtain $\braket{x|C'|x}$ simply by multiplying the expressions $\braket{x|G|x}$ for different gates $G$ in $C'$. Each gate corresponds to a term in $f_C$ as defined above. So, for all $x \in \{0,1\}^\ell$, $\braket{x|C'|x} = (-1)^{f_C(x)}$, and hence
\[ \braket{0|H^{\otimes \ell}C'H^{\otimes \ell}|0} = \frac{1}{2^\ell} \sum_{x \in \{0,1\}^\ell} \braket{x|C'|x} = \frac{1}{2^\ell} \sum_{x \in \{0,1\}^\ell} (-1)^{f_C(x)} = \frac{\gap(f_C)}{2^\ell}. \]
We can remove any Hadamard gates in $C'$ using a trick from~\cite{bremner11}. Imagine we have a Hadamard gate on the $i$'th qubit. We form a new overall circuit $C''$ from $C$ by introducing a new ancilla qubit $a$ initialised in the state $\ket{0}$, replacing the Hadamard gate with the gadget $G = H_i CZ_{ai} H_a$, and changing all subsequent gates involving the $i$'th qubit to use qubit $a$ (see Figure \ref{fig:posts} for an illustration). Then, by direct calculation, $\bra{0}_i G \ket{0}_a = H / \sqrt{2}$, so $\braket{0|C''|0} = \braket{0|C|0}/\sqrt{2}$. Following this procedure for each of the $h$ Hadamard gates in $C'$, we obtain a circuit on $n = \ell+h$ qubits, where each Hadamard gate corresponds to a product of two variables and relabelling of a qubit as specified in the definition of $f_C$. Taking into account the normalisation factor of $2^{h/2}$, we obtain
\[ \bracket{0}{C}{0} = \frac{1}{2^{h/2+\ell}} \sum_{x \in \{0,1\}^n} (-1)^{f_C(x)} = \frac{\gap(f_C)}{2^{h/2+\ell}} \]
as claimed.
\end{proof}

\begin{figure}
\[
\Qcircuit @C=1em @R=.7em {
\lstick{\dots} & \gate{U} & \gate{H} & \gate{V} & \rstick{\dots} \qw
}
\;\;\;\;\;\;\;\;\;\;
\mapsto
\;\;\;\;\;\;\;\;\;\;
\raisebox{0.5cm}{
\Qcircuit @C=1em @R=.7em {
\lstick{\dots} & \gate{U} & \ctrl{1} & \gate{H} & \rstick{\bra{0}} \qw  \\
\lstick{\ket{0}} & \gate{H} & \control \qw & \gate{V} & \rstick{\dots} \qw
}
}
\]
\caption{Replacing a Hadamard gate with a controlled-Z gate and postselection~\cite{bremner11}.}
\label{fig:posts}
\end{figure}
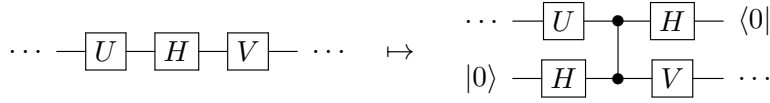

It is easy to check that the formula of Proposition \ref{prop:gap} is accurate for the example in Figure \ref{fig:labcirc} (where $\gap(f_C)=16$ and $\bracket{0}{C}{0} = 1/2$). The correspondence between circuits and polynomials given in Proposition \ref{prop:gap} will be the main tool used throughout this paper. We remark that all the other amplitudes $\braket{x|C|y}$, $x,y \in \{0,1\}^\ell$, are also related to polynomials. This is because X gates inserted at the start or end of $C$ can be used to map $\ket{0} \mapsto \ket{y}$ or $\ket{x} \mapsto \ket{0}$, X gates can be commuted through Hadamard gates to produce Z gates, and Z gates give linear terms in the corresponding polynomial. Thus $\braket{x|C|y} = \gap(f_C + L_{x,y})/2^{h/2+\ell}$ for some linear function $L_{x,y}$ depending on $x$, $y$.

We next make some other simple observations that follow from the circuit-polynomial correspondence.

% ------------------------------------------------------------------------------

\subsection{Basic observations}

\begin{obs}
There can be more than one quantum circuit $C$ corresponding to a given polynomial $f_C$.
\end{obs}

\begin{proof}
There are two easy ways to see this. First, as Z, CZ and CCZ gates commute, a consecutive sequence of such gates in $C$ can be reordered arbitrarily while still corresponding to the same polynomial $f_C$. Second, it is sometimes the case that CZ gates and Hadamards are interchangeable. For example, Figure \ref{fig:samecircuits} shows two circuits which both correspond to the polynomial $x_1 x_2$.
\end{proof}

\begin{figure}
\[
\begin{array}{m{1in}m{1in}}
\Qcircuit @C=1em @!R {
& \ustick{x_1} \qw & \ctrl{1} \qw & \qw & \qw \\
& \ustick{x_2} \qw & \control \qw & \qw & \qw\\
}
&
\;\;\;\;
\Qcircuit @C=1em @R=.7em {
& \ustick{x_1} \qw & \gate{H} \qw & \ustick{x_2} \qw & \qw
}
\end{array}
\]
\caption{The internal part of two circuits which both correspond to the polynomial $x_1 x_2$.}
\label{fig:samecircuits}
\end{figure}
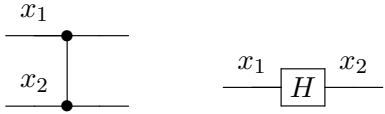

\begin{obs}
\label{obs:iqp}
For every degree-3 polynomial $f:\{0,1\}^n \rightarrow \{0,1\}$ with no constant term, there exists a quantum circuit $C$ on $n$ qubits such that $f = f_C$.
\end{obs}

\begin{proof}
Produce the internal part of a circuit $C$ on $n$ qubits by associating a qubit with each variable in $f$, and include a Z, CZ or CCZ gate between the qubits corresponding to each degree 1, 2, 3 term (respectively) in $f$.
\end{proof}

We remark that the class of quantum circuits produced from the procedure in Observation \ref{obs:iqp} are IQP circuits~\cite{shepherd09}. An IQP circuit (``Instantaneous Quantum Polynomial-time'') on $n$ qubits is a circuit of the form $H^{\otimes n}DH^{\otimes n}$, where $D$ is a circuit of $\poly(n)$ diagonal gates. It was argued in~\cite{bremner15} that it should be hard to sample classically from the output probability distributions of quantum circuits of the form of Observation \ref{obs:iqp}, even up to small total variation distance. The argument was based on a plausible complexity-theoretic conjecture regarding the complexity of approximately computing $\gap(f)$ for random degree-3 polynomials $f$.

\begin{obs}
\label{obs:largew}
There exists a degree-3 polynomial $f:\{0,1\}^n \rightarrow \{0,1\}$ such that every quantum circuit $C$ corresponding to $f$ requires $n$ qubits.
\end{obs}

\begin{proof}
Consider the polynomial containing the term $x_i x_j x_k$ for all $1 \le i<j<k \le n$, and no other terms. As there are no degree-2 terms, any corresponding circuit $C$ cannot contain any internal Hadamard gates. Thus $C$ must act on at least $n$ qubits, with one qubit corresponding to each variable.
\end{proof}

\begin{obs}
If $f_C:\{0,1\}^n \rightarrow \{0,1\}$ corresponds to a quantum circuit $C$ on $\ell$ qubits, then $|\gap(f_C)| \le 2^{n/2 + \ell/2}$.
\end{obs}

\begin{proof}
From Proposition \ref{prop:gap}, $\braket{0|C|0} = \gap(f_C) / 2^{h/2+\ell}$. As $\braket{0|C|0}$ is a quantum circuit amplitude and hence bounded by 1 in absolute value by unitarity, $|\gap(f_C)| \le 2^{h/2+\ell} = 2^{n/2+\ell/2}$.
\end{proof}

These observations motivate us to define the {\em quantum circuit width} $w(f)$ of a degree-3 polynomial $f$ over $\F_2$ as the minimal number of qubits required for any quantum circuit which corresponds to $f$. For example, the family of polynomials $f$ in Observation \ref{obs:largew} has $w(f) = n$, whereas the polynomial $f' = x_1 x_2 + x_2 x_3 + \dots + x_{n-1} x_n$ has $w(f') = 1$, corresponding to a circuit whose internal part consists of $n-1$ Hadamard gates applied to one qubit.

% ------------------------------------------------------------------------------

\section{Computational complexity}
\label{sec:compcomp}

The theory of computational complexity studies the inherent difficulty of computational problems. One of the main goals of this field is to classify problems into complexity classes: sets of problems of comparable difficulty. We now give a brief, informal introduction to this area; see~\cite{papadimitriou94,arora09} for a full, formal treatment. The complexity classes used in this work can all be presented in terms of determining properties of classical or quantum circuits. A classical circuit is a collection of AND, OR and NOT gates connected with wires, which map an input to an output by evaluating the gates in the natural manner. We assume that classical circuits only have one output bit, but potentially many input bits. For each classical circuit $C$, we let $C(x)$ be the output of $C$ given the bit-string $x$ as input. Then we can define the following natural problems:
\begin{itemize}
%\item {\sc Circuit Value}: given a classical circuit $C$ and an input to the circuit $x$, determine whether $C(x) = 1$.
\item Circuit SAT: given a classical circuit $C$, determine whether there exists $x$ such that $C(x) = 1$.
\item Circuit Counting: given a classical circuit $C$, output $|\{x:C(x)=1\}|$.
\end{itemize}
Each of these problems corresponds to a complexity class.
%The class P (``polynomial-time'') is the class of decision problems which reduce to the {\sc Circuit Value} problem.
NP (``nondeterministic polynomial-time'') is the class of decision problems which reduce to Circuit SAT in polynomial time, while \#P (``sharp-P'' or ``number-P'') is the class of functional problems which can be expressed as an instance of Circuit Counting. The closely related class P$^{\#\text{P}}$ is the class of functional problems which can be solved in polynomial time, given the ability to solve any problem in the class \#P. For example, the problem of computing $|\{x:C_1(x)=1\}|-|\{x:C_2(x)=1\}|$ for circuits $C_1$, $C_2$ is in P$^{\#\text{P}}$. Here ``polynomial time'' is short for ``in time polynomial in the input size'', which is the key notion of efficiency used in computational complexity. For any complexity class $\mathcal{C}$, a problem $\mathcal{P}$ is said to be $\mathcal{C}$-hard if it is at least as hard as every problem in $\mathcal{C}$: in other words, for every problem in $\mathcal{C}$, there is a polynomial-time reduction from that problem to $\mathcal{P}$.

A problem is said to be NP-complete if it is equivalent in difficulty to Circuit SAT, up to polynomial-time reductions. Many important practical problems (such as optimal packing and scheduling, integer programming, and computing ground-state energies of classical physical systems) are known to be NP-complete~\cite{garey79}. The famous P$\stackrel{?}{=}$NP problem effectively asks whether Circuit SAT can be solved in time polynomial in the size of the given circuit. Although it is widely believed that the answer is ``no'', a positive answer would have momentous consequences, implying that any NP-complete problem could be solved in polynomial time. Observe that Circuit Counting is at least as hard as Circuit SAT. In fact, it is conjectured that this problem is much harder. Indeed, if there existed an efficient reduction from Circuit Counting to Circuit SAT, then the infinite tower of complexity classes known as the polynomial hierarchy would collapse~\cite{toda91}, a consequence similar to P$=$NP and considered almost as unlikely.

Many interesting problems in physics and elsewhere are known to be \#P-hard: at least as hard as any problem in \#P. These include computing Ising model partition functions~\cite{jerrum93}, evaluating Jones and Tutte polynomials~\cite{jaeger90}, and exactly computing the permanent of a 0-1 matrix~\cite{valiant79}. The intuitive reason behind the hardness of these problems is that they involve computing a sum of exponentially many terms. However, surprisingly, in some cases such sums can be computed efficiently (exactly or approximately). Examples include exact computation of Ising model partition functions on planar graphs~\cite{fisher61,kasteleyn63,temperley61}, approximate computation of the permanent of a non-negative matrix~\cite{jerrum04}, and Valiant's quantum-inspired ``holographic algorithms'' for combinatorial problems~\cite{valiant08a}. Proving \#P-hardness of a problem provides strong evidence that a clever efficient algorithm like these should not exist for that problem.

% ------------------------------------------------------------------------------

\subsection{Computational complexity of low-degree polynomials}
\label{sec:phard}

We can use the connection between quantum circuits and polynomials to prove \#P-hardness results. It was shown by Ehrenfeucht and Karpinski~\cite{ehrenfeucht90} that computing the number of zeroes (equivalently, the gap) of a degree-3 polynomial $f$ over $\F_2$ is \#P-hard. This implies that using the circuit-polynomial correspondence is unlikely to give an efficient algorithm for simulating all quantum circuits classically by computing quantum circuit amplitudes. However, we can go in the other direction, and use the correspondence to obtain a quantum proof of \#P-hardness of computing the number of zeroes of $f$ (equivalently, computing $\gap(f)$).

\begin{prop}
\label{prop:phard}
It is \#P-hard to compute $\gap(f)$ for degree-3 polynomials $f$.
\end{prop}

\begin{proof}
We will show that the problem of exactly computing $\braket{0|C|0}$ for an arbitrary quantum circuit $C$ containing Hadamard, Z, CZ, and CCZ gates is \#P-hard. As computing $\gap(f)$ for arbitrary degree-3 polynomials $f$ would allow us to compute $\braket{0|C|0}$ for arbitrary circuits of this form, this will imply the claim. To achieve this, we first show that computing $\braket{0|C|0}$ for an arbitrary quantum circuit $C$ containing Hadamard, X and Toffoli gates is \#P-hard. This can easily be obtained from a similar result of Van den Nest~\cite{vandennest08}; we include a simple direct proof here for completeness.

It is a fundamental result in the theory of reversible computation that X and Toffoli gates together with ancillas are universal for classical computation, i.e.\ that given a boolean function $g:\{0,1\}^n \rightarrow \{0,1\}$ computed by a classical circuit $C$ of $\poly(n)$ gates, there is a quantum circuit $C'$ of $\poly(n)$ X and Toffoli gates such that $C'\ket{x}_I\ket{0}_O\ket{0}^{\otimes a}_A = \ket{x}_I\ket{g(x)}_O\ket{0}^{\otimes a}_A$, where the circuit acts on a Hilbert space divided into an $n$-qubit input register I, a 1-qubit output register O, and an $a$-qubit ancilla register A. Then let the circuit $C''$ be defined as follows:
\begin{enumerate}
\item Apply an X gate to the O register.
\item Apply Hadamard gates to each qubit in the I and O registers.
\item Apply $C'$.
\item Apply Hadamard gates to each qubit in the I and O registers.
\item Apply an X gate to the O register.
\end{enumerate}
If $C''$ is applied to the initial state $\ket{0}$, the state prepared after the second step is $\ket{+}^{\otimes n}_I \ket{-}_O \ket{0}^{\otimes a}_A$. When $C'$ is applied in the third step the second and third registers are left unchanged, and the state of the first register becomes
\[ \ket{\psi_g} = \frac{1}{\sqrt{2^n}} \sum_{x \in \{0,1\}^n} (-1)^{g(x)} \ket{x}. \]
Thus $\braket{0|C''|0} = \braket{+|^{\otimes n} | \psi_g} = \frac{1}{2^n} \sum_{x \in \{0,1\}^n} (-1)^{g(x)} = \gap(g)/2^n$. So computing $\braket{0|C''|0}$ allows us to determine $\gap(g)$, and hence the number of zeroes of $g$, for functions $g$ computed by arbitrary polynomial-size classical circuits. This problem is \#P-hard by definition.

It remains to show that this same conclusion holds for circuits containing Hadamard, Z, CZ, and CCZ gates. But this is immediate, as Toffoli gates can be produced from CCZ gates by conjugating the target qubit by a Hadamard, and similarly $X = HZH$.
\end{proof}

%As shown by Ehrenfeucht and Karpinski~\cite{ehrenfeucht90}, computing the number of zeroes (equivalently, the gap) of a degree-3 polynomial is \#P-hard, so this does not appear to give us an efficient algorithm for simulating quantum circuits. (In fact, one can see the above construction as a {\em proof} that it is \#P-hard to compute $\gap(f)$! Using the classical universality of Toffoli gates, it is easy to construct a quantum circuit $C$ containing only Hadamard and Toffoli gates such that $\bracket{0}{C}{0}$ encodes a \#P-hard problem~\cite{vandennest08}, and CCZ gates are equivalent to Toffolis up to a Hadamard gate on the target qubit.)

The \#P-hardness proof of Ehrenfeucht and Karpinski~\cite{ehrenfeucht90} is not difficult. However, the quantum proof gives a different perspective, and also lends itself to simple generalisations. For example:

\begin{prop}
\label{prop:3terms}
$\gap(f)$ remains \#P-hard to compute for degree-3 polynomials where each variable appears in at most 3 terms.
\end{prop}

\begin{proof}
We show that computing $\gap(f)$ for an arbitrary degree-3 polynomial $f$ reduces to computing $\gap(f')$ for a degree-3 polynomial $f'$ where each variable appears in at most 3 terms. Given $f$, we produce a corresponding quantum circuit $C$. Then, between each pair of gates, we insert two Hadamard gates on each qubit to produce a new circuit $C'$. As $H^2 = I$, $\braket{0|C'|0} = \braket{0|C|0}$, so the corresponding polynomial $f_{C'}$ satisfies $\gap(f_{C'}) = \gap(f_C)$, up to an easily computed scaling factor. But each variable in $f_{C'}$ is only contained within at most 3 terms, because the inserted Hadamard gates effectively relabel all the variables between each pair of terms in the polynomial.
%So, for each term in $f_C$ (e.g.\ $x_1 x_2 x_3$), each variable in that term appears in at most 3 terms in $f_{C'}$ (e.g.\ $x_1$ only appears in terms of the form $x_1 y_1 + x_1 x_2 x_3 + x_1 z_1$).
\end{proof}

A similar circuit simplification to that of Proposition \ref{prop:3terms} was previously observed in~\cite{rudolph09}. Proposition \ref{prop:phard} shows that we should not hope to find an efficient algorithm for simulating arbitrary quantum circuits by computing the number of zeroes of low-degree polynomials. However, for some classes of polynomials we can indeed obtain efficient algorithms (see below for some examples of this).

A natural question is whether we can improve Proposition \ref{prop:phard} to show that even computing the number of zeroes of degree-2 polynomials is \#P-hard. It was already shown by Ehrenfeucht and Karpinski~\cite{ehrenfeucht90} that this is unlikely to be the case, as there is a polynomial-time algorithm for this problem. There is an alternative ``quantum'' way of seeing this result, as relating to ideas around the well-known Gottesman-Knill theorem~\cite{nielsen00}, which states that any quantum circuit whose gates are all picked from the Clifford group can be efficiently simulated classically. Indeed, for any degree-2 polynomial $f:\{0,1\}^n \rightarrow \{0,1\}$, by Observation \ref{obs:iqp} we can write down a quantum circuit $C$ on $n$ qubits containing only Hadamard, Z and CZ gates such that $\braket{0|C|0} = \gap(f)/2^n$. As the gates in $C$ are all members of the Clifford group, the state $C\ket{0}$ is a stabilizer state, as is the state $\ket{0}$. It is known that the inner product between two arbitrary stabilizer states can be computed in time $O(n^3)$~\cite{aaronson04a,garcia14,bravyi16}, implying an $O(n^3)$ algorithm for computing $\gap(f)$ for degree-2 polynomials $f:\{0,1\}^n \rightarrow \{0,1\}$.

% ------------------------------------------------------------------------------

\subsection{Approximate computation}
\label{sec:approxcomplexity}

Given that we have shown exactly computing $\gap(f)$ to be hard, the next natural question is whether we can approximately compute it. We now show that this question is closely connected to {\em quantum} computational complexity. The class of decision problems which can be solved efficiently by a quantum computer (i.e.\ in time polynomial in the size of the input), with success probability $2/3$, is known as BQP~\cite{watrous09}. As with the classical complexity classes discussed previously, BQP can be expressed in terms of circuits; however, the circuits are now quantum. Any polynomial-time quantum computation solving a decision problem can be expressed as applying some quantum circuit $U$, generated from the input in polynomial time, to the initial state $\ket{0}$, then measuring the first qubit, and returning the measurement result.

\begin{prop}
\label{prop:gapbqp}
Determining $\gap(f)$ for arbitrary degree-3 polynomials $f:\{0,1\}^n \rightarrow \{0,1\}$ up to absolute error $\frac{1}{3} \cdot 2^{(n+w(f))/2}$ is BQP-hard.
%The following problem is PromiseBQP-complete. Given $c$ such that $1/\poly(n) \le c < 1$, a degree-3 polynomial $f:\{0,1\}^n \rightarrow \{0,1\}$ and a description of a minimal width quantum circuit representing $f$, determine $\gap(f)$ up to additive error $c 2^{(n+w(f))/2}$.
\end{prop}

\begin{proof}
We first recall that solving decision problems reduces to computing quantum circuit amplitudes (this is an observation of Knill and Laflamme~\cite{knill98}). Assume that we are given some quantum circuit $C$ containing only Hadamard, Z, CZ and CCZ gates, which is applied to the initial state $\ket{0}$, followed by a measurement of the first qubit. We would like to approximately determine the probability that this measurement outputs 1. As the set of gates $\{$Hadamard, CCZ$\}$ is universal for quantum computation~\cite{shi03,aharonov03}, this is sufficient to solve any problem in BQP. So consider the following circuit $C'$:
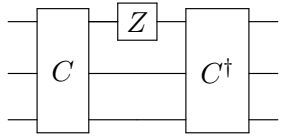
\[
\Qcircuit @C=1em @R=.7em {
 & \multigate{2}{C} & \gate{Z} \qw & \multigate{2}{C^\dag} & \qw  \\
 & \ghost{C} & \qw & \ghost{C^\dag} & \qw \\
 & \ghost{C} & \qw & \ghost{C^\dag} & \qw 
}
\]
Then $\bracket{0}{C'}{0} = \bracket{0}{C^\dag Z_1 C}{0} = \tr Z_1 (C\ket{0}\bra{0}C^\dag)$, which is precisely the difference between the probability that the measurement outputs 0, and the probability that it outputs 1. By the definition of the error bounds in BQP, we have $|\bracket{0}{C'}{0}| \ge 1/3$, so it is sufficient to estimate $\bracket{0}{C'}{0}$ up to absolute error less than $1/3$ to determine whether the answer should be 0 or 1. As discussed in Section \ref{sec:circpoly}, we can assume that $C'$ begins and ends with Hadamards on every qubit (equivalently, that $C$ begins with Hadamards on every qubit).

From Proposition \ref{prop:gap}, there is a degree-3 polynomial $f_{C'}:\{0,1\}^n \rightarrow \{0,1\}$, where $n=h+\ell$, $h$ is the number of Hadamard gates in the internal part of $C'$ and $\ell$ is the number of qubits on which $C'$ acts, such that $\braket{0|C'|0} = \gap(f_{C'}) / 2^{h/2+\ell}$. So it is sufficient to determine $\gap(f_{C'})$ up to absolute accuracy $\frac{1}{3} \cdot 2^{h/2+\ell} = \frac{1}{3} \cdot 2^{n/2+\ell/2}$ to solve the original decision problem. Observing that $\ell \ge w(f)$ by definition completes the proof.
\end{proof}

We have seen that approximately computing $\gap(f)$ up to accuracy $O(2^{(n+w(f))/2})$ is sufficient to simulate arbitrary quantum computations. This is already sufficient to imply the known complexity class inclusion\footnote{In fact, this argument also gives an alternative proof of the tighter complexity class inclusion BQP$\subseteq$AWPP, due to Fortnow and Rogers~\cite{fortnow98}.} BQP$\subseteq$P$^{\#\text{P}}$~\cite{bernstein97,dawson05}, as it is easy to see that $\gap(f)$ can be computed exactly by counting the number of inputs of a circuit which evaluate to 1, and hence is in P$^{\#\text{P}}$. Does the implication go the other way? That is, can we use quantum computation to approximate $\gap(f)$ up to accuracy $O(2^{(n+w(f))/2})$? If so, this would imply that approximating $\gap(f)$ up to this level of accuracy is effectively equivalent\footnote{Technically, equivalent to the complexity class PromiseBQP~\cite{janzing07}: the class of problems which reduce to determining whether the acceptance probability of a quantum computation is greater than $2/3$ or less than $1/3$, given the promise that exactly one of these is the case.} to the complexity class BQP. This would give a new example of a combinatorial problem which characterises the power of quantum computation. Several such examples are known (e.g.~\cite{knill01,aharonov06,janzing07,vandennest08a}), but approximately computing the number of zeroes of degree-3 polynomials would arguably be the simplest yet.

For any quantum circuit $C$, the Hadamard test~\cite{aharonov06} can be used to estimate $\braket{0|C|0}$ up to inverse-polynomially small absolute error. So, if we are given a circuit on $\ell$ qubits corresponding to a polynomial $f$, we can estimate $\gap(f)$ up to accuracy $O(2^{n/2+\ell/2})$. If $\ell = w(f)$, we have achieved an approximation which matches the bound of Proposition \ref{prop:gapbqp}. However, it is not clear how to efficiently determine a quantum circuit corresponding to $f$ which acts on $w(f)$ qubits. Indeed, even determining $w(f)$ itself could be NP-complete.

\begin{prob}
What is the complexity of computing $w(f)$ for an arbitrary degree-3 polynomial $f:\{0,1\}^n \rightarrow \{0,1\}$?
\end{prob}

To achieve a good enough level of accuracy in estimating $\gap(f)$, it would be sufficient to find a circuit on $\ell$ qubits such that $\ell = w(f) + O(\log n)$. But it is non-obvious how to obtain even this level of accuracy.

We can also relate the quantum circuit width to the complexity of {\em classical} simulation.

\begin{prop}
\label{prop:classical}
Given a degree-3 polynomial $f:\{0,1\}^n \rightarrow \{0,1\}$ and a description of a quantum circuit on $\ell$ qubits corresponding to $f$, $\gap(f)$ can be calculated exactly classically in time $O(2^{2\ell} \poly(n))$. Further, $\gap(f)$ can be approximated up to additive error $\epsilon\,2^n$ with success probability $2/3$ in time $O(\poly(n)/\epsilon^2)$.
\end{prop}

\begin{proof}
For any quantum circuit $C$ on $\ell$ qubits containing $m$ gates, $\braket{0|C|0}$ can be calculated in time $O(2^{2\ell} m)$ simply by multiplying out the matrices. If $C$ represents $f$, it can be assumed to contain at most $\poly(n)$ gates, so $m = \poly(n)$. For the second part, we can estimate $|\{x:f(x)=0\}|/2^n$ by taking the average of $s$ random samples from $f(x)$. Each sample can be computed in time $\poly(n)$. By a standard Chernoff bound argument~\cite{dubhashi09}, in order for this estimate to be correct up to absolute error $\epsilon$ with probability $2/3$, it is sufficient to take $s = O(1/\epsilon^2)$.
\end{proof}

Using the second approach in Proposition \ref{prop:classical}, we can achieve the same level of approximation accuracy achieved by an optimal quantum circuit by taking $\epsilon = O(2^{(w(f)-n)/2})$, giving a classical algorithm which runs in time $O(2^{n-w(f)} \poly(n))$. Thus observe that, if either $w(f) \ge n - O(\log n)$ or $w(f) \le O(\log n)$, the speedup we could obtain by using a quantum algorithm to compute $\gap(f)$ cannot be super-polynomial (but apparently for different reasons). In the former case, the approximate classical algorithm from Proposition \ref{prop:classical} runs in polynomial time; in the latter case, the exact classical algorithm runs in polynomial time.
%But even width $n/2$ (say) might give a non-trivial algorithm?

These results motivate us to further explore the concept of quantum circuit width.

% ------------------------------------------------------------------------------

\subsection{Quantum circuit width}
\label{sec:width}

We first show that most degree-3 polynomials $f$ have high quantum circuit width, and hence that $\gap(f)$ cannot be approximated significantly more efficiently using this quantum circuit approach than is possible classically.

\begin{prop}
The probability that a random degree-3 polynomial $f:\{0,1\}^n \rightarrow \{0,1\}$ with no constant term has $w(f) \le n-3$ is at most $2^{(-3n+1)/2}$.
\end{prop}

\begin{proof}
We count the number of different functions which can correspond to a circuit on $k$ qubits of the form discussed in this work whose internal part contains $n-k$ Hadamards (giving a polynomial on $n$ variables). Break the internal part of the circuit into $n-k+1$ horizontal blocks such that each Hadamard $H_1,\dots,H_h$ begins a block. Then slide (commute) all the Z, CZ, CCZ gates in the circuit to the left until they cannot go any further (i.e.\ come up against a Hadamard). Then, except for the furthest left-hand block, each such gate acts on the qubit corresponding to the Hadamard which begins its block. Therefore, there are at most $2^{\binom{k}{2} + k + 1} = 2^{k(k+1)/2+1}$ different possibilities for the combination of gates in each block, except the left-hand block, where there are $2^{\binom{k}{3} + \binom{k}{2} + k} = 2^{k(k^2+5)/6}$ possibilities. There are $k^{n-k}$ possibilities for the vertical position of the Hadamards. Overall, we get an upper bound on the number of functions that can be produced which is equal to
\[ 2^{(n-k)(k(k+1)/2+1) + k(k^2+5)/6 + (n-k) \log_2 k}. \]
Take the rough upper bound $\log_2 k \le k/2$, valid for large enough $k$. Then the above quantity is increasing with $k$ and for $k=n-3$ is equal to $2^{(n^3-4n+3)/6}$. On the other hand, there are $2^{\binom{n}{3}+\binom{n}{2}+n} = 2^{(n^3+5n)/6}$ degree-3 polynomials on $n$ variables with no constant term. Thus the fraction of polynomials $f$ such that $w(f) \le n-3$ is at most exponentially small in $n$.% (Interestingly, for $k=n-2$, this argument no longer works! We get an upper bound of $2^{(n^3+5n-6)/6}$.)
\end{proof}

We next relate the quantum circuit width of a polynomial to a combinatorial parameter of a hypergraph associated with the polynomial. A hypergraph $G = (V,E)$ is defined by a set of vertices $V$ and a set of hyperedges $E$, where a hyperedge is a subset of at least 2 of the vertices. We can associate a degree-3 polynomial $f$ with a hypergraph $G(f)$ by associating each variable with a vertex, and thinking of each term involving at most 3 variables as a hyperedge between at most 3 vertices. A proper $k$-colouring of a hypergraph $G$ is an assignment of colours to vertices, picked from a set of colours of size $k$, such that at least two vertices within each hyperedge are assigned different colours. The chromatic number of a hypergraph, $\chi(G)$, is defined to be the minimal $k$ such that there exists a proper $k$-colouring of $G$.

\begin{prop}
For any degree-3 polynomial $f:\{0,1\}^n \rightarrow \{0,1\}$, $\chi(G(f)) \le 2 w(f)$, and this inequality can be tight. However, there exists a family of polynomials $f:\{0,1\}^n \rightarrow \{0,1\}$, for $n$ even, such that $\chi(G(f)) = 2$ but $w(f)=n/2$.
\end{prop}

\begin{proof}
Given a circuit for $f$ using $\ell$ qubits, each pair of variables which are associated with the same qubit but are not adjacent cannot be included in the same term of $f$. We can thus properly colour the vertices of $G(f)$ using at most $2\ell$ colours by associating a pair of colours $(c_i,d_i)$ with each qubit, and allocating colour $c_i$ (resp.\ $d_i$) to those vertices which occur on line $i$ at odd (resp.\ even) times. Tightness follows from the function $f(x) = x_1 x_2 + x_2 x_3 + \dots + x_{n-1} x_n$, which has $w(f)=1$. As $G(f)$ is a path on $n$ vertices, $\chi(G(f))=2$.

For the second part, consider the polynomial $f(x) = x_1 x_2 + x_3 x_4 + \dots + x_{n-1} x_n$. The corresponding graph $G(f)$ consists of $n/2$ disjoint edges and hence can be properly coloured with 2 colours.
\end{proof}

% ------------------------------------------------------------------------------

\section{Polynomials and simulation of quantum circuits}
\label{sec:polysim}

We have seen that, using the construction of Proposition \ref{prop:gapbqp}, in order to simulate a quantum circuit -- i.e.\ to determine the probability that, at the end of the circuit, the result of measuring the first qubit would be 1 -- it is sufficient to compute $\gap(f)$ for a related function $f$. One can use this idea to easily obtain various simulation results for classes of quantum circuits.

First, as discussed in Section \ref{sec:approxcomplexity}, any circuit containing only Hadamard, Z and CZ gates can be simulated efficiently classically using the Gottesman-Knill theorem~\cite{nielsen00}. This result can be generalised to circuits containing a small number of CCZ gates as follows.

\begin{prop}
Let $S$ be a hitting set for the collection of degree-3 terms of $f:\{0,1\}^n \rightarrow \{0,1\}$ (in other words, a set of variables such that each degree-3 term contains at least one element of $S$). Then, given $f$ and $S$, $\gap(f)$ can be computed in time $O(2^{|S|} \poly(n))$.
\end{prop}

\begin{proof}
For any variable $x_i$, let $f_{x_i\leftarrow z}$ denote the function obtained from $f$ by fixing the value of $x_i$ to $z$. Then it is easy to see that $\gap(f) = \gap(f_{x_i\leftarrow 0}) + \gap(f_{x_i\leftarrow 1})$. Applying this recursively, for any set $S$ of variables, $\gap(f)$ can be computed by summing the gaps of the $2^{|S|}$ functions obtained by fixing each of the variables in $S$ to either 0 or 1. If we choose $S$ to include at least one variable from each of the degree-3 terms in $f$, each new polynomial produced has degree at most 2, and hence has gap computable in time $O(n^3)$~\cite{ehrenfeucht90,aaronson04a}.
\end{proof}

Observe that, if $f$ contains $k$ degree-3 terms, there is always a hitting set containing $k$ elements (just by taking one variable from each term). More generally, we would like to find a hitting set of minimal size $h(f)$. This is an NP-complete problem~\cite{garey79}, but luckily an approximation $h'(f) \le 3h(f)$ can be found in polynomial time (approximating $h(f)$ any better than this is NP-hard~\cite{khot08}, assuming the Unique Games Conjecture from complexity theory). We therefore have that $\gap(f)$ can be computed in time $2^{O(h(f))} \poly(n)$.

The construction of Proposition \ref{prop:gapbqp} produces a circuit $C'$ from any circuit $C$ on $\ell$ qubits whose corresponding polynomial $f_{C'}$ satisfies $h(f_{C'}) \le 2 h(f_C)$. Therefore, any polynomial-size circuit $C$ can be simulated in time $2^{O(h(f_C))} \poly(\ell)$. It was already shown by Aaronson and Gottesman that circuits on $\ell$ qubits containing $k$ non-Clifford gates can be simulated in time $2^{O(k)} \poly(\ell)$~\cite{aaronson04a}, and more recent work has improved the constant hidden in the $O(k)$ term for circuits where the only non-Clifford gate is the T gate~\cite{bravyi16,bravyi16a}. However, the result here is somewhat more general in that there exist circuits with many CCZ gates whose corresponding polynomial has a small hitting set. For example, Figure \ref{fig:treew} illustrates a circuit on $\ell$ qubits containing CCZ gates from the first qubit to every other pair of qubits; this circuit has $\binom{\ell-1}{2}$ gates but a hitting set of size 1.

Also observe that this simulation does not seem to follow immediately from the results of Markov and Shi~\cite{markov08} on simulating quantum circuits by tensor contraction in time exponential in the tree-width of the circuit. Indeed, there exist circuits that contain only Clifford gates but have arbitrarily high tree-width.

\begin{figure}
\[
\Qcircuit @C=1.2em @R=1.2em {
 & \ctrl{2} & \ctrl{3} & \ctrl{4} & \qw & \dots & & \ctrl{7} \qw & \qw \\
 & \control \qw & \control \qw & \control \qw & \qw & \dots & & \qw & \qw\\
 & \control \qw & \qw & \qw & \qw & \dots & & \qw & \qw \\
 & \qw & \control \qw & \qw & \qw & \dots & & \qw & \qw \\
  & \qw & \qw & \control \qw & \qw & \dots & & \qw & \qw \\
 & & \vdots & & & & \vdots \\
 & \qw & \qw & \qw & \qw & \dots & & \control \qw & \qw \\
 & \qw & \qw & \qw & \qw & \dots & & \control \qw & \qw
}
\]
\caption{A circuit with a hitting set of size 1 but many non-Clifford gates.}
\label{fig:treew}
\end{figure}
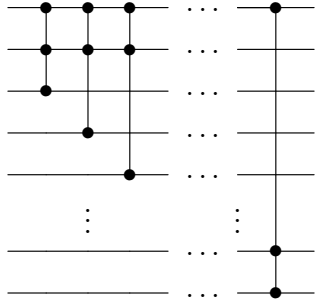

% ------------------------------------------------------------------------------

\subsection{Simulation by linear transformations}

In order to calculate $\gap(f)$ more efficiently, we can attempt to transform $f$ into a polynomial which is simpler in some sense. One way of doing this is to apply a linear transformation to $f$. The following result is well-known in the theory of error-correcting codes~\cite{macwilliams83}; we include the simple proof for completeness.

\begin{prop}
For any degree-3 polynomial $f:\{0,1\}^n \rightarrow \{0,1\}$, and any nonsingular linear transformation $L \in GL_n(\F_2)$, let $f^L$ be the polynomial $f^L(x) = f(Lx)$. Then $\deg(f^L) = 3$ and $\gap(f^L) = \gap(f)$.
\end{prop}

\begin{proof}
To produce $f^L$ from $f$, we can replace each term $x_i x_j x_k$ with a term $(Lx)_i (Lx)_j (Lx)_k$ (and similarly for the terms dependent on 1 or 2 variables). As $(Lx)_i$ is a linear function of $x$ over $\F_2$, and similarly for $j$, $k$, the product of these functions is a polynomial of degree at most 3. For the second part, as $L$ is nonsingular, there is a one-to-one mapping between the set $\{x:f(x) = 0\}$ and the set $\{x:f^L(x) = 0\}$, so $\gap(f^L) = \gap(f)$.
\end{proof}

In fact, the group $GL_n(\F_2)$ is known to be the {\em largest} group of transformations which preserves polynomial degree~\cite{macwilliams83}. In some cases, a linear transformation can completely change a function's quantum circuit width and hence the efficiency with which its gap can be computed using the exact algorithm of Proposition \ref{prop:classical}. As a very simple example, it is easy to show that the polynomial $x_1 + \dots + x_n$ has quantum circuit width $n$, but following a linear transformation that maps $x_1 + \dots + x_n \mapsto x_1$, the resulting polynomial $x_1$ has quantum circuit width 1.

Although we do not know a general way of minimising the quantum circuit width of a function by applying a linear transformation, a simpler approach is to minimise the number of variables on which the function depends. Given a polynomial $f$ which depends on $v$ variables, $\gap(f)$ can be computed exactly in time $O(2^v \poly(v))$ simply by evaluating $f$ on each of the $2^v$ possible assignments to the variables. It has been shown by Carlini~\cite{carlini06} (see also Appendix B of~\cite{kayal11}) that the linear transformation $L$ which minimises the number of variables in $f^L$ can be computed in polynomial time. We therefore obtain the following corollary:

\begin{cor}
Let $C$ be a polynomial-size quantum circuit on $\ell$ qubits such that there exists a linear transformation $L$ such that $f_C^L$ depends on $v$ variables. Then there is a classical algorithm which computes $\braket{0|C|0}$ exactly in time $O(2^v \poly(\ell))$.
\end{cor}

In particular, if there exists $L$ such that $f_C^L$ depends on $O(\log \ell)$ variables, we obtain a polynomial-time classical simulation of $C$.

% ------------------------------------------------------------------------------

\section{Conclusions}
\label{sec:conclusions}

In this work we have investigated a correspondence between quantum circuits and low-degree polynomials over finite fields, and have shown that by exploiting this correspondence we can obtain classical hardness results, as well as ideas for classical algorithms that simulate quantum circuits. There seem to be many interesting directions in which to further explore this area. For example, as discussed in Section \ref{sec:approxcomplexity}, what is the complexity of computing or approximating the quantum circuit width $w(f)$? Is it related to other measures of complexity of boolean functions? Low-degree polynomials over $\F_2$ are equivalent to Reed-Muller codes~\cite{macwilliams83} -- can ideas from classical coding theory be applied to understand quantum circuits? And finally, can any other useful simulation techniques be developed by taking this perspective -- perhaps for other specific classes of quantum circuits?

% ------------------------------------------------------------------------------

\subsection*{Acknowledgements}

This work was supported by an EPSRC Early Career Fellowship (EP/L021005/1). Some of this work was carried out while the author was at the University of Cambridge. I would like to thank Mick Bremner and Dan Shepherd for discussions on this topic over the last few years, and Scott Aaronson, Miriam Backens and Richard Jozsa for helpful comments on a previous version. Special thanks to Sophie for arriving safely, and providing many helpful distractions from completing this work.

% ------------------------------------------------------------------------------
% ------------------------------------------------------------------------------

\bibliographystyle{plain}
\bibliography{../../thesis}

\end{document}